\documentclass{amsart}

\newtheorem{theorem}{Theorem}[section]
\newtheorem{lemma}[theorem]{Lemma}
\newtheorem{proposition}{Proposition}[section]
\usepackage{amsmath,amssymb,amsthm}
\theoremstyle{definition}
\newtheorem{definition}[theorem]{Definition}
\newtheorem{example}[theorem]{Example}

\newtheorem{corollary}[theorem]{Corollary}
\theoremstyle{remark}
\newtheorem{remark}[theorem]{Remark}
\usepackage{enumitem}
\usepackage{amsmath,amssymb,amsthm}
\usepackage{mathrsfs}
\numberwithin{equation}{section}
\title{Block Codes on Pomset Metric}

\author{Atul Kumar Shriwastva}
\address{{Department of Mathematics, National Institute of Technology Warangal, Hanamkonda, Telangana 506004, India}}
\email{shriwastvaatul@student.nitw.ac.in}
\thanks{}

\author{R. S. Selvaraj}
\address{{Department of Mathematics, National Institute of Technology Warangal, Hanamkonda, Telangana 506004, India}}
\email{rsselva@nitw.ac.in}
\date{\today}
\begin{document}

\begin{abstract}
  Given a regular multiset $M$ on $[n]=\{1,2,\ldots,n\}$, a partial order $R$ on $M$,  and a label map  $\pi : [n] \rightarrow \mathbb{N}$ defined by $\pi(i) = k_i$ with $\sum_{i=1}^{n}\pi (i) = N$, we define a pomset block metric $d_{(Pm,\pi)}$ on the direct sum $ \mathbb{Z}_{m}^{k_1} \oplus \mathbb{Z}_{m}^{k_2} \oplus \ldots \oplus \mathbb{Z}_{m}^{k_n}$ of  $\mathbb{Z}_{m}^{N}$ based on the pomset $\mathbb{P}=(M,R)$. The  pomset block metric extends the classical pomset metric introduced  by I. G. Sudha and  R. S. Selvaraj and generalizes the poset block metric introduced by M. M. S. Alves et al over $\mathbb{Z}_m$.	
  The space  $ (\mathbb{Z}_{m}^N,~d_{(Pm,\pi)} ) $ is called the pomset block space and we determine the complete weight distribution of it. Further, $I$-perfect pomset block codes for ideals with partial and full counts are described. 
  Then, for block codes with chain pomset, the packing radius and Singleton bound are established. The relation between MDS codes and $I$-perfect codes for any ideal $I$ is investigated. Moreover, the duality theorem for an MDS pomset block code is established when all the blocks have the same size.
\end{abstract}
	\maketitle
\section{Introduction}	
The major problem of coding theory is to find the largest minimum distance $d$ of any $k$-dimensional linear code of length $n$ for any integer $ n > k \geq 1$. 
Numerous researchers have worked on the structure of posets and defined various spaces such as poset space \cite{Bru} including crown space \cite{crown} and hierarchical poset space \cite{hierarchical}, RT space \cite{rt} including NRT space \cite{nrt classification} and so on. They investigated the properties such as packing radius, Singleton bound, maximum distance separability, weight distribution, and perfectness of codes with those spaces. 
Over the past two decades, the study of block codes has sparked several significant developments in the communication field, such as experimental design, high-dimensional numerical integration, and cryptography.
In 2006,  $\pi$-block codes of length $N$ over $\mathbb{F}_q$ were introduced by K. Feng et al. \cite{fxh}, which is another generalization of codes with Hamming  metric.  
Further, it is extended to $(P,\pi)$-block codes by M. M. S. Alves et al. \cite{Ebc}, by introducing a partial order relation on the block positions $[n]$ of an $N$-tuple in $\mathbb{F}_q^N$ consisting of $n$-blocks in each $N$-tuple, thereby providing an extension to the poset codes \cite{Bru} as well. Thus, it kept the researchers to concentrate on the exploration of block codes with various metrics, allowing one to study the class of posets (such as hierarchical posets, NRT posets, etc.) that admit certain properties of the block codes (\cite{Ebc},\cite{bk},\cite{bkd},\cite{bkdnsr},\cite{fxh},\cite{hkmdspc},\cite{nrt classification},\cite{MacAdmitting classification}). This has emerged as a useful research topic in the digital communication world. 
  \par A space equipped with pomset metrics is a recently introduced one by I. G. Sudha, R. S. Selvaraj \cite{gsrs}, which is a generalization of Lee space \cite{Leecode}, in particular, and poset space \cite{Bru}, in general, over $\mathbb{Z}_m$.
  Further, the MacWilliam type identities were determined  \cite{gsr},  maximum distance separability, and perfectness of  codes \cite{gr} were studied on codes with pomset metric. A pomset is a partially ordered multiset, and for more information, one can see \cite{BWD},  \cite{GKJ}, \cite{GK} and \cite{gsrs}. Of late, L. Panek \cite{wcps} introduced the notion of weighted coordinates poset metric, a general form of the pomset metric without using the multiset structure. 
  But, as the notion of ideals $I$ with partial count and those with full count in pomsets affect the study of $I$-perfect codes and hence MDS codes, it is felt by the authors to pursue the exploration of block codes by retaining the flavour of multiset set up. 
  \par In poset space or poset block spaces, the $I$-balls are linear, which is not the case with pomset spaces. There are two kinds of ideals in pomsets \cite{gr}: ideals with full count and ideals with partial count. So, we extend the concept of $I$-balls ($I$ is an ideal in the pomset) to the pomset block space $\mathbb{Z}_m^N$ to study the $I$-perfect pomset codes. We will see that if $I$ is an ideal with a full count then the $I$-ball is linear, whereas this is not the case for ideals with a partial count. Thus, there is a need to explore the properties of $I$-balls for ideals with a  partial count. Further, we investigate the $I$-perfect pomset block codes for both ideals with partial count and full count  (see Theorem \ref{I-perfect full pomset code} and \ref{I-perfect partial pomset code }). However, these $I$-perfect codes are more comfortable to study than the $r$-perfect codes because the $I$-balls  are submodules of $\mathbb{Z}_m^N$. 
\par This paper introduces the pomset block metric for codes of length $N $ over the ring $\mathbb{Z}_m$. Section ${2}$ provides the preliminaries and basic results on multisets, relations on multisets, and pomset. In Section ${3}$, we define the $(Pm, \pi)$-metric (or pomset block metric) on $\mathbb{Z}_{m}^{N}$. Further, we define $r$-balls, $I$-balls and explore their structures. We determine the complete weight distribution of $(Pm,\pi)$-space. 
An $I$-perfect pomset block code for an ideal  $I$ with full count is found.
 But, an $I$-perfect pomset block code $\mathbb{C}$ for an ideal $I$ with a partial count is constructed by imposing certain conditions on the subset $\mathbb{C}$  of $\mathbb{Z}_m^{N}$. Section ${4}$ completely deals with block codes with chain pomsets. There we obtain the packing radius and Singleton bound for pomset block codes. Then the relationship of MDS poset block codes with MDS pomset block codes is found. Moreover, the relationship of MDS $(Pm, \pi)$-codes with $I$-perfect $(Pm, \pi)$-codes is studied and the duality theorem for such $(Pm, \pi)$-codes when all blocks have the same length is established. 
\section{Preliminaries}

A multiset $M$  (in short, mset) is a collection of elements that may contain duplicates. These elements are drawn from a set $A$ and the multiset $M$ is described through a count function
 $C_M : A  \rightarrow \mathbb{N}$ such that $C_M (a) = c$ for $a \in A$.
  For example, if $A=\{ a_1, a_2, \ldots, a_n\}$ then $ M=\{c_1/a_1, c_2/a_2, \dots, c_n/a_n \}$ is a multiset drawn from the set $ A $ where  $ c_i/a_i \in M $  represents an element $ a_i \in A $ appearing $ c_i $ times in $M$. 
The cardinality of a multiset $M$ is  $|M| = \sum\limits_{a \in A} C_M (a) $. For a $h \in \mathbb{N}$, if $c_i=h$ for every $i$  then $M$ is called a regular mset with height $h$.  Given a multiset $M$, we call the set $ M^{*} \triangleq \{i \in A : C_M (i) > 0 \}$ to be the root set of $M$. 
Let $M_1$ and $M_2$ be two msets drawn from the set $A$. If $C_{M_1} (a) \leq C_{M_2} (a) $ for all $a \in A$, then $M_1$ is called as submultiset (in short, submset) of $M_2$ $(M_1 \subseteq M_2)$ and in addition, if $C_{M_1} (a) < C_{M_2} (a)$ for an $a \in M_2^*$ then $M_1$ is said to be a proper submset of $M_2$ $(M_1 \subset M_2)$. $M_1 = M_2$ iff $ M_1 \subseteq M_2$ and $M_2 \subseteq M_1$. Union and intersection of multisets $M_1$ and $M_2$ are defined as:
  $ M_1 \cup M_2 \triangleq \{ C_{M_1 \cup M_2} (a)/a : C_{M_1 \cup M_2} (a) =  \max\{C_{M_1} (a), C_{M_2} (a)\} \text{ for all } a \in A\}$ and
    $  M_1 \cap M_2 \triangleq \{ C_{M_1 \cap M_2} (a)/a : C_{M_1 \cap M_2} (a) =  \min \{C_{M_1} (a), C_{M_2} (a)\} \text{ for all } a \in A\}$, respectively. 
\par    The mset space $[A]^r$ is the set of all multisets $M$ drawn from $A$ such that $C_{M} (a) \leq r $ for every $a \in A$. Let $M_1,M_2 \in [A]^r$, then their mset sum
$  M_1 \oplus M_2 \triangleq \{ C_{ M_1 \oplus M_2} (a)/a : C_{ M_1 \oplus M_2} (a) = \min \{C_{M_1} (a)+C_{M_2} (a), r\} \text{ for all } a \in A\}$. The mset difference of $M_2$ from $M_1 $  is 
$  M_1 \ominus M_2 \triangleq \{ C_{ M_1 \ominus M_2} (a)/a : C_{ M_1 \ominus M_2} (a) = \max \{C_{M_1} (a)-C_{M_2} (a), 0\} \text{ for all } a \in A\}$. 
\sloppy{Cartesian product:  $M_1 \times M_2 =  \{mn/(m/a,n/b) : m/a \in M_1 \text{ and }  n/b \in M_2\}$. The notation $k/(m/a, n/b) $ means that $a $ is appearing $ m $ times in $M_1$, $b$ is appearing $ n $ times in $M_2$ and the pair $(a, b)$ is appearing $k$ times in $M_1 \times M_2$ where $1 \leq  k \leq mn$. 
	The submset $R$ of $M \times M $ is said to be an mset relation on $M$
	if for every $(m/a,n/b) \in R$ one has $C_{R} (m/a,n/b) = C_{M} (a) \cdot C_{M} (b)$.  For $(m/a,n/b) \in R$, the  count of $a$ and  $b$ in the ordered pair $(m/a,n/b)$ is denoted as $C_{M} (a)=m $ and $ C_{M} (b) =n$ respectively. }
\par A mset relation $R \subseteq M \times M$ is said to be partially ordered mset  relation (in short, pomset relation) if it satisfies the following: (1) for every $p/a \in M$, $p/aR p/a$ (reflexive), (2)  $p/a R q/b$ and $q/b R p/a$ imply $p=q, a=b$ (anti-symmetric), and  (3) $p/a R q/b$ and $q/b R t/c$ imply $p/a R t/c$ (transitive).  The pair $(M, R)$ is called a partially ordered multiset (or pomset) and is denoted by $\mathbb{P}$. An element $q/b \in M$ is said to be a maximal element of $\mathbb{P}$ if there is no element $t/c \in M$ such that $q/b R t/c$. An element $r/a \in M$ is said to be a minimal element of $\mathbb{P}$ if there is no element $t/c \in M$ such that $t/c R r/a$. $\mathbb{P}$ is said to be a chain if every distinct pair of elements from $M$ is comparable in $\mathbb{P}$. $\mathbb{P}$ is said to be an antichain if no distinct pair of elements from $M$ is comparable in $\mathbb{P}$.
\par A submset $I$ of $M$ is called an  ideal of $\mathbb{P}$ if $p/a \in I$ and $q/b ~R ~ p/a~(b \neq a)$ imply $q/b \in I$. An ideal generated by an element $p/a \in M$ is defined as $\langle p/a \rangle = \{p/a\} \cup \{q/b \in M : q/b R p/a \} $. An ideal generated by a submset $I$ of $M$ is defined as $\langle I \rangle = \bigcup\limits_{p/a \in I} \langle p/a \rangle $.  An ideal $I$ is said to be of full count if $C_{I}(i)=C_{M}(i)$ for every $i \in I^*$; otherwise, it is said to be an ideal with a partial count. Throughout the paper, $\mathscr{I}(\mathbb{P})$ denotes the set of all ideals in $\mathbb{P}$
and $\mathscr{I}^t(\mathbb{P})$ denotes the set of all ideals in $\mathbb{P}$ with cardinality $t$.
\begin{example}\label{example}
   Let $M = \{ h/i_j : j \in [5] \}$ be a regular mset of height $h$ and let $ R = \{h^2/(h/i_j, h/i_j),  h^2/(h/i_1, h/i_j) : j \in [5] \} $ be a pomset relation defined on $M$. Then $P = (M,R) $ is a pomset. Some of the ideals of $P = (M,R) $ are: (a) $I_\ell= \langle \{h_1/i_\ell\} \rangle = \{h/i_1, h_1/i_\ell \} $ for any $l \in \{2,3,4,5\}$ and $h_1 \leq h$; (b) $I= \{ h/i_1,h/i_2,h_1/i_3 \}$; (c)  $J = \{h/i_1, h/i_2, h/i_3, h/i_4,h_1/i_5\}$. If $h_1 = h$ then $I_\ell$, $I$ and $J$ are ideals with full counts. If $h_1 < h$ then these are ideals of $\mathbb{P}$  with partial counts.
   	But, if we take $ R' =	\{h^2/(h/i_j, h/i_j),  h_2/(h/i_1, h/i_k) : j,k \in \{1,2,3,4,5 \}, k \neq 1 \text{ and } h_2 < h^2\} $ then $R'$ is not an mset relation because the counts of the ordered pairs $(h/i_1, h/i_k)$ are $h_2 < h^2$. Hence,	$	\mathbb{P} = (M,R') $ is not a  pomset.
\end{example}
\begin{proposition}\label{p2}
	Let $\mathbb{P} = (M, R)$ be a pomset where $M \in [A]^r$. If  $I$, $J \in \mathscr{I}(\mathbb{P})$, then the following holds: 
	\begin{enumerate}[label=(\roman*)]
	\item   $\langle I \cup J \rangle = \langle I \rangle \cup \langle J \rangle$ and $\langle I \cap J \rangle = \langle I \rangle \cap \langle J \rangle$.
	\item\label{p2b}  $\langle I \oplus J \rangle \subseteq \langle I \rangle \oplus \langle J \rangle $ if M is a regular mset with height $r$. 
	\end{enumerate}
\end{proposition}
\begin{proposition}\label{J in I will exist}
	Let $\mathbb{P} = (M, R)$ be a pomset and $0 \leq m \leq n \leq |M|$. 
	Then  for each $I \in \mathscr{I}^n(\mathbb{P})$ there exists a $J \in \mathscr{I}^m(\mathbb{P}) $ such that $J \subseteq I$. Moreover, for each $J \in \mathscr{I}^m(\mathbb{P})$ there exists an $I \in \mathscr{I}^n(\mathbb{P}) $ such that $I \supseteq J$.
\end{proposition}
 \begin{definition}[Dual Pomset]
 	For a pomset $ \mathbb{P}  = (M,R)$,  its dual pomset $\tilde{\mathbb{P} }= (M, \tilde{R})$ is defined with the same underlying mset $M $ of height $r$ such that $p/a \tilde{R} q/b $ in $ \tilde{\mathbb{P}}$ if and only if $q/b R p/a$ in  $ \mathbb{P} $. As a result, the order ideals of $\tilde{\mathbb{P} }$ are precisely the complements of the order ideals of $\mathbb{P}$, that is, $\mathscr{I}(\tilde{ \mathbb{P} }) = \{I^c : I \in \mathscr{I}({\mathbb{P} }) \}$, where the compliment
   of $I$ with respect to  $M$ is $  I^c \triangleq \{ C_{I^c} (i)/i : C_{I^c} (i) = r - C_{I}(i) \text{ for all } i \in M^* \}$.
 \end{definition}
\begin{example}
   For the pomset  $ \mathbb{P}$ given in Example \ref{example}, the dual is $\tilde{\mathbb{P} }= (M, \tilde{R})$ with the multiset relation $\tilde{R} = \{h^2/(h/i_j, h/i_j),  h^2/(h/i_k, h/i_1) : j,k \in \{1,2,3,4,5 \} $. The maximal elements of the pomset $\mathbb{P}$ viz., $ h/i_2, h/i_3,h/i_4, \text{ and } h/i_5$ are the minimal elements of its dual $\tilde{\mathbb{P} }$. 
\end{example}
\section{Pomset Block Metric Space ($(Pm,\pi)$-space)}
For the regular multiset $ M=\{\lfloor \frac{m}{2} \rfloor/i : i \in [n]\}$ of height $ \lfloor \frac{m}{2} \rfloor$ drawn from $ [n] $ with a partial order $R$, the pair  $ \mathbb{P} = (M,R)$ is a pomset. For any positive integer $k$ and $ x=(x_{1}, x_{2}, \dots, x_{k}) \in  \mathbb{Z}_{m}^{k}$, the support of $ x $ with respect to Lee weight is defined to be the multiset $ supp_{L}(x)=\{c/i : x_i \neq 0,  c=w_{L}(x_{i}) \}$  
where $ w_{L}(x_{i})=\min\{x_{i}, m-x_{i}\} $ is the Lee weight of $ x_{i} \in  \mathbb{Z}_{m} $.
 Let  $\pi$ be a label map from $[n] $ to $ \mathbb{N}$  defined by $\pi(i) = k_i$ with $\sum_{i=1}^{n}\pi (i) = N$ and consider the space $ \mathbb{Z}_m^N $ as the direct sum of modules $ \mathbb{Z}_{m}^{k_1}, \mathbb{Z}_{m}^{k_2}, \ldots, \mathbb{Z}_{m}^{k_n} $; that is,  $\mathbb{Z}_{m}^{N} = \mathbb{Z}_{m}^{k_1}  \oplus \mathbb{Z}_{m}^{k_2} \oplus \ldots \oplus \mathbb{Z}_{m}^{k_n} $.  Every $N$-tuple  $v$ in $ \mathbb{Z}_{m}^{N} $ is expressed uniquely as $	v = v_1 \oplus v_2 \oplus \ldots \oplus v_n $ where
   $v_i = (v_{i_1},v_{i_2},\ldots,v_{i_{k_i}}) \in \mathbb{Z}_{m}^{k_i}$.  
  For $v_i \in \mathbb{Z}_{m}^{k_i} $, if $ Max_c supp_{L}(v_i) \triangleq \max\limits_{i_j \in supp_L(v_i)^*} \{w_{L}(v_{i_j})\}  $ denote  the maximum among the Lee weights of the components of $v_i$, we define the pomset block support or $(Pm,\pi)$-support of  $ v  \in \mathbb{Z}_{m}^{N} $  as 
\begin{align*}
	supp_{(Pm,\pi)}(v) = \{ c_i/i \in M : v_i \neq 0 ~\text{and} ~c_i = Max_c supp_{L}(v_i)   \}
\end{align*}
   a submset of $M$. 
  \par  Subsequently, the pomset block weight or $ (Pm,\pi) $-weight of $v$ is  $ w_{(Pm,\pi)}(v) \triangleq |\langle supp_{(Pm,\pi)}(v) \rangle| $ and the pomset block distance or $ (Pm,\pi) $-distance between  $ u, v \in \mathbb{Z}_{m}^N $ is    $d_{(Pm,\pi)}(u,v)  \triangleq  w_{(Pm,\pi)}(u - v)$.
   Now,  the $ (Pm,\pi) $-distance is indeed  a metric on $ \mathbb{Z}_{m}^N $ as given in:
\begin{theorem}\label{t1}
	The pomset block distance $d_{(Pm,\pi)}(.,.)$ is  a metric on $\mathbb{Z}_{m}^N $.
\end{theorem}
\begin{proof}
	\sloppy{Let $ u, v, w \in \mathbb{Z}_{m}^N $. As $|supp_{(Pm,\pi)}(u-v)| \geq 0 $ so is  $d_{(Pm,\pi)}(u,v)  \geq 0 $. 
	Clearly, $d_{(Pm,\pi)}(u, v) = 0 $ iff $u=v$.  
	As $ w_L(-v_{i_\ell}) = w_L(v_{i_\ell}) $ for each $i$ and $\ell$, so is 
	 $supp_{L}{(-v)} = supp_{L}{(v)}$ and hence $ w_{(Pm,\pi)}(-v) = w_{(Pm,\pi)}(v) $. Thus, $ d_{(Pm,\pi)}(u, v) = d_{(Pm,\pi)}(v, u)$. 
	Since $Max_c supp_{L}(u_i + v_i) \leq Max_c supp_{L}(u_i) + Max_c supp_{L}(v_i) $,  then $supp_{(Pm, \pi)}(u + v) \subseteq supp_{(Pm, \pi)}(u) \oplus supp_{(Pm,\pi)}(v) $. It follows  from the Proposition \ref{p2}\ref{p2b} that, $\langle supp_{(Pm, \pi)}(u + v) \rangle \subseteq  \langle supp_{(Pm, \pi)}(u) \oplus supp_{(Pm, \pi)}(v) \rangle \subseteq \langle supp_{(Pm, \pi)}(u) \rangle \oplus  \langle supp_{(Pm, \pi)}(v) \rangle  $. Thus, $ w_{(Pm,\pi)}(u + v) = |\langle supp_{(Pm, \pi)}(u + v) \rangle| \leq w_{(Pm,\pi)}(u) + w_{(Pm,\pi)}(v)$. Hence, $d_{(Pm,\pi)}(u,v) \leq d_{(Pm,\pi)}(u,w) + d_{(Pm,\pi)}(w, v)$.}
\end{proof}
The metric $d_{(Pm,\pi)}(.,.)$ on $ \mathbb{Z}_{m}^N $ is called as the pomset block metric or $ (Pm,\pi) $-metric. The pair $ (\mathbb{Z}_{m}^N,~d_{(Pm,\pi)} ) $ is said to be a pomset block space (or $(Pm,\pi) $-space). 	
\par If $k_i = 1 $ for every $i\in [n]$ then  $Max_c supp_{L}(v_i)=w_{L}(v_i)$ and  
$		supp_{(Pm,\pi)}(v)  = \{ k/i \in M : v_i \neq 0 ~\text{and} ~k = w_{L}(v_i)\} 	= supp_{L}(v)$. Here, the pomset block weight becomes pomset weight and the  pomset block space  becomes the  pomset space \cite{gsrs}. 
\begin{example} \label{ex1}
	Consider $\mathbb{Z}_{7}^{18}=\mathbb{Z}_{7}^2 \oplus \mathbb{Z}_{7}^3 \oplus \mathbb{Z}_{7}^4 \oplus \mathbb{Z}_{7}^4 \oplus \mathbb{Z}_{7}^3 \oplus \mathbb{Z}_{7}^2 $ with $N=18$ and $n=6$.  Let $\preceq $ be a partial order relation  on $[6]=\{1,2,3,4,5,6\}$ such that $3/1 \preceq 3/2 \preceq 3/4$; $3/5 \preceq 3/6 $. Then the set 
  $	A=	\{x=x_1 \oplus x_2 \oplus \bar{0} \oplus 0101 \oplus x_5 \oplus 20: x_1 \in \mathbb{Z}_{7}^2,  ~x_2 \in \mathbb{Z}_{7}^3, ~x_5 \in \mathbb{Z}_{7}^3  \}$ 
	gives a set of vectors in $\mathbb{Z}_{7}^{18}$ in which $(Pm,\pi)$-weight of any $x \in A$ is  $ |\langle supp_{(Pm,\pi)}(x) \rangle|= |\langle \{1/4,2/6\} \rangle|=3+3+1+3+2=12$. 
\end{example}
\par Any subset $\mathbb{C} $ of a $(Pm,\pi) $-space is said to be a pomset block code or a $(Pm,\pi)$-code. The minimum distance of $\mathbb{C} $ is $d_{(Pm,\pi)}\mathbb{(C)} \triangleq \min \{  d_{(Pm, \pi)} {(c_1, c_2)}: c_1, c_2 \in \mathbb{C} \}$. 
If $\mathbb{C}$ is linear, 
$\centering d_{(Pm,\pi)}\mathbb{(C)} = min \{ w_{(Pm,\pi)}(c) : 0 \neq c \in \mathbb{C} \} $.
As  $ w_{(Pm,\pi)}(v) \leq n \lfloor \frac{m}{2} \rfloor $ for any $v \in \mathbb{Z}_{m}^N $, the  minimum distance of  any $(Pm,\pi)$-code $\mathbb{C} $ is bounded above by $ n \lfloor \frac{m}{2} \rfloor $. 
\subsection{Balls and $I$-perfect codes in $(Pm,\pi)$-space }
\par 
 The pomset block ball or $(Pm,\pi)$-ball with center $ u \in \mathbb{Z}_{m}^N $ and radius $r$ (or $r$-ball centered at $u$) is defined as
 $B_{(Pm,\pi)}(u, r) = \{ v \in \mathbb{Z}_{m}^N : d_{(Pm,\pi)}(u, v) \leq r \}$ whereas an	$r$-sphere centered at $u$ is defined as $S_{(Pm,\pi)}(u, r) = \{ v \in \mathbb{Z}_{m}^N : d_{(Pm,\pi)}(u, v) = r \}$. If $u=0$, they will also be denoted as $B_r$ and $S_r$ respectively.

 \par	Given a pomset $\mathbb{P}=(M,R)$ on a regular multiset $ M$ of height $ \lfloor \frac{m}{2} \rfloor$ drawn from  $[n]$, we define a poset $P=([n], \preceq_P)$ \textit{induced by  $\mathbb{P}$} to be such that $ i \preceq_P j$ in $P$ iff $ p/i R q/j $ in  $\mathbb{P}$.

\par It is easy to see that if $I$ is an ideal in pomset $\mathbb{P}$ then $I^*$ is an ideal in the poset $P$ induced by $\mathbb{P}$.
 Now, for a $u \in \mathbb{Z}_{m}^N $, its $\pi$-support is $ supp_{\pi}(u) =\{ i \in [n] : u_i \neq 0\}$ and for an ideal $I^*$ in $P$, the $I^*$-ball centered at $u$ with respect to $(P,\pi)$-metric is $B_{I^*}(u) \triangleq \{ v \in \mathbb{Z}_{m}^N :  supp_{\pi}(u- v)  \subseteq I^* \}$. The $r$-ball centered at $u$ with respect to $(P,\pi)$-metric is  $B_{(P,\pi)}(u,r) \triangleq \{ v \in \mathbb{Z}_{m}^N :  d_{(P,\pi)}(u- v)  \leq r \}$.  Let $\mathcal{I}^r(P)$ denotes the set of all ideals in $P$ of cardinality $r$. Let $ d_{(P,\pi)}(\mathbb{C})$ be the minimum distance of a code $\mathbb{C} $ of length $N$ over $\mathbb{Z}_m$ with respect to $(P,\pi)$-metric. 
 \par 
 Since $ |\langle supp_{\pi}(v) \rangle| \leq w_{(Pm,\pi)}(v) \leq  \lfloor\frac{ m }{2}\rfloor  |\langle supp_{\pi}(v) \rangle| $ $\forall$ $v \in \mathbb{Z}_m^N $, we have $ d_{(P,\pi)}(u,v) \leq d_{(Pm,\pi)}(u,v) \leq  \lfloor\frac{ m }{2}\rfloor d_{(P,\pi)}(u,v) $ for all $u, v \in \mathbb{Z}_m^N $. 
 Then, $ d_{(P,\pi)}(\mathbb{C}) \leq d_{(Pm,\pi)}(\mathbb{C}) \leq  \lfloor\frac{ m }{2}\rfloor d_{(P,\pi)}(\mathbb{C}) $.
 Thus, we have the following:
\begin{proposition} \label{phg}
	Let $d_{(Pm,\pi)}(\mathbb{C})$  and  $d_{(P,\pi)}(\mathbb{C})$ be the minimum distances of a code $\mathbb{C} $ of length $N $ over $\mathbb{Z}_m$, with respect to $(Pm,\pi)$-metric and $(P,\pi)$-metric respectively.  Then $	\big\lfloor \frac{d_{(Pm,\pi)} \mathbb{(C)}-1}{ \big\lfloor\frac{ m }{2}\big \rfloor}\big\rfloor \leq d_{(P,\pi)}(\mathbb{C}) - 1$.	
\end{proposition} 
 \par  Let $r= t\lfloor\frac{ m }{2}\rfloor+ s $, $1 \leq s \leq \lfloor\frac{ m }{2}\rfloor$ and $1 \leq t \leq n-1$. If $v \in B_{(P,\pi)}(u, t) $ then $ supp_{(P,\pi)}(u-v) \subseteq I$ for an $I \in \mathcal{I}^{t}(P)$. Then, $ w_{(Pm,\pi)}(u -v) \leq t \lfloor\frac{ m }{2}\rfloor$. Thus, we have 
 \begin{theorem}
 	For $x \in \mathbb{Z}_{m}^N $ and $1 \leq t \leq n-1$, $B_{(P,\pi)}(x, t) \subseteq B_{(Pm,\pi)}(x, t \lfloor\frac{ m }{2}\rfloor)$.
 \end{theorem}
 
\begin{definition}[{$r$-perfect $(Pm, \pi)$-code}]
	 A $(Pm, \pi)$-code $\mathbb{C}$ of length $N$ over ${\mathbb{Z}_m}$ is said to be $r$-perfect if the $r$-balls centered at the codewords of $C$ are pairwise disjoint and their union covers the entire space $\mathbb{Z}_{m} ^ N$.
\end{definition} 
\par For a submset $I$ in $\mathbb{P}$, the $I$-ball centered at $u$ is 
$B_{I}(u) \triangleq \{ v \in \mathbb{Z}_{m}^N :  supp_{(Pm,\pi)}(u- v)  \subseteq I \}$ and
 the $I$-sphere centered at $u$ is $S_{I}(u) \triangleq \{ v \in \mathbb{Z}_{m}^N :  supp_{(Pm,\pi)}(u-v) = I \}$.
For $v \in  B_I (u)$, it is not necessary that $\langle supp_{(Pm,\pi)} {(u - v)} \rangle \subseteq  I$. 
If the submset $I $ is an ideal in $\mathbb{P}$ then $\langle supp_{(Pm,\pi)} {(u - v)} \rangle \subseteq  I$ will always hold. Hence, for an ideal $I$ in $\mathbb{P}$, the $I$-ball centered at $u$ can also be defined as  
$B_{I}(u) \triangleq \{ v \in \mathbb{Z}_{m}^N : \langle supp_{(Pm,\pi)}(u- v) \rangle \subseteq I \}$ and $d_{(Pm,\pi)}(u,v) \leq |I|$ for each $ v \in B_{I}(u)$. Moreover, the $I$-sphere centered at $u$ is $S_{I}(u) \triangleq \{ v \in \mathbb{Z}_{m}^N : \langle supp_{(Pm,\pi)}(u-v) \rangle = I \}$. In short, the $B_I$ and $S_I$ denote the $I$-ball and $I$-sphere centered at $0 \in \mathbb{Z}_{m}^N$ respectively. 
 Let $ \mathscr{I}_j^i ( \mathbb{P} ) $ denote the collection of all ideals with cardinality $ i $ having $ j $ maximal elements. Then $
 \bigcup\limits_{j=1}^{\min\{i,n\}} \mathscr{I}_j^i (\mathbb{P})= \mathscr{I}^i{(\mathbb{P})}$. Let $Max(I) = \{c_{i_1}/i_1,c_{i_2}/i_2,\ldots,c_{i_j}/i_j \}$ be the submset of maximal elements in the ideal $I$.  
 \begin{proposition}
 	Let $I$ be a submset (or ideal) in $\mathbb{P}$ and $v \in \mathbb{Z}_m^N$. Then $S_{I}(v)$ is the translate (or coset) of $S_I$ containing $v$, that is, $S_I(v) = v + S_I$. Moreover, $B_I(v) = v + B_I$.
 \end{proposition}
 \begin{proof}
 	If $u \in S_I(v)$ then $ supp_{(Pm,\pi)} {(u - v)} = I$ so that $u - v\in S_I $ and thus, $u  = v + (u-v) \in v +S_I$. Let $y \in S_I$. Then  $ supp_{(Pm,\pi)} {(v - (v+y))} = 
 	supp_{(Pm,\pi)} {(-y)} = supp_{(Pm,\pi)} {(y)} =I$.
 	Thus, $  v + y \in S_I(v) $ for all $y \in S_I$. Hence, $S_I(v) = v + S_I$. Then, $B_I(v)=v + B_I$ follows.
 \end{proof}
 \begin{remark}
 	In a similar way as above, for $v \in \mathbb{Z}_m^N$, an $r$-sphere and an $r$-ball centered at $v$ can also be expressed as $S_{(Pm,\pi)}(v, r) = v + S_r$ and $B_{(Pm,\pi)}(v, r)= v + B_r$ respectively.
\end{remark}
\begin{proposition}
	Two vectors $u$ and $v$ belong to the same $I$-ball for some $I \in \mathscr{I}^r(\mathbb{P})$ if and only if $d_{(Pm,\pi)}(u,v) \leq r$.
\end{proposition} 
\begin{proposition}\label{runion ball}
	Every $r$-ball is the union of all $I$-balls where $I \in \mathscr{I}^r(\mathbb{P})$, that is,	$	B_{(Pm,\pi)}(u,r)= \bigcup\limits_{I \in \mathscr{I}^r(\mathbb{P})}B_{I}(u)$.
\end{proposition} 
Given an ideal $I$, for each $i_t \in Max(I)^*$ with $C_I(i_t)=C_{i_t}$, we set 
\begin{equation*}
	N_{i_t}  =  \left\{
	\begin{array}{ll}
		2   & \hbox{if $c_{i_t} < \frac{m}{2} $} \\
		1   & \hbox{if $c_{i_t} = \frac{m}{2} $.} 
	\end{array} \right.
\end{equation*}	
\begin{theorem}\label{S_I cardinality}
	Let $I$ be an ideal in $\mathbb{P}$. Then the cardinality of $I$-sphere centered at $ 0 \in \mathbb{Z}_{m}^N $ is 
	\begin{align*}
		|S_I | =\prod\limits_{i_t \in Max(I)^*} \delta(c_{i_t} )  \prod\limits_{l \in {I \setminus Max(I)}} m^{k_l}.
	\end{align*}
	where $c_{i_t} = C_I(i_t)$ and $\delta(c_{i_t} ) =  \left\{
	\begin{array}{ll}
		(2c_{i_t}+1)^{k_{i_t}} - (2c_{i_t}-1)^{k_{i_t}} & \hbox{if $c_{i_t} \neq \lfloor \frac{m}{2} \rfloor $ } \\
		m^{k_{i_t}} - (m-1)^{k_{i_t}}  & \hbox{if $c_{i_t} = \lfloor \frac{m}{2} \rfloor $.}
	\end{array} \right. $ 
\end{theorem}
\begin{proof}
	Let  $Max(I) = \{c_{i_1}/i_1,c_{i_2}/i_2,\ldots,c_{i_j}/i_j \}$.  Let  $S_I$ be the $I$-sphere centered at zero. 
	Choose $y =y_{1} \oplus y_{2} \oplus \ldots \oplus y_{n} \in \mathbb{Z}_{m}^N$ such that if $i_t \in Max(I)^*$ then $c_{i_t} \neq 0$ and $Max_c{supp_{L}{y_{i_t}}} =  c_{i_t}$; and if  $i_t \in I^* \setminus Max(I)^*$ then $y_{i_t}  \in \mathbb{Z}_{m}^{k_{i_t}}$; and  if $i_t \notin I^*$ then $y_{i_t} = 0$. Thus, $y \in S_I$. \\
	Case(1)  If $i_t \in  Max(I)^*$ is with full count $\lfloor \frac{m}{2} \rfloor$ in the ideal $I$, the total number of choices for such $y_{i_t}= (y_{i_{t_{1}}},y_{i_{t_{2}}},\ldots,y_{i_{t_{k_{i_t}}}})$ in  $\mathbb{Z}_{m}^{k_{i_t}}$ with $Max_c{supp_{L}{y_{i_t}}} = \lfloor \frac{m}{2} \rfloor$
	is \begin{align*}
		&= {k_{i_t} \choose 1}N_{i_t}(m - N_{i_t})^{k_{i_t}-1} +{k_{i_t} \choose 2}N_{i_t}^{2}(m - N_{i_t})^{k_{i_t}-2}+ \ldots  + {k_{i_t} \choose k_{i_t}}N_{i_t}^{k_{i_t}} \\
		&=m^{k_{i_t}}-(m - N_{i_t})^{k_{i_t}}
	\end{align*}
	Case(2) ~If $i_t \in  Max(I)^*$ is with partial count say $ c_{i_t} < \lfloor \frac{m}{2} \rfloor$ in the ideal $I$, the total number of choices for such $y_{i_t}= (y_{i_{t_{1}}},y_{i_{t_{2}}},\ldots,y_{i_{t_{k_{i_t}}}})$ in  $\mathbb{Z}_{m}^{k_{i_t}}$ with $Max_c{supp_{L}{y_{i_t}}} = c_{i_t} < \lfloor \frac{m}{2} \rfloor$
	is \begin{align*}
		&= {k_{i_t} \choose 1}2(2c_{i_t}-1)^{k_{i_t}-1} +{k_{i_t} \choose 2}2^{2}(2c_{i_t}-1)^{k_{i_t}-2}+ \ldots  + {k_{i_t} \choose k_{i_t}}2^{k_{i_t}} \\
		&=(2c_{i_t}+1)^{k_{i_t}}- (2c_{i_t}-1)^{k_{i_t}}
	\end{align*}
	Hence the result.
\end{proof}
Since $S_I(u)=u + S_I$, $|S_r|=\sum\limits_{j=1}^{min\{i, n\}} \sum\limits_{I \in {\mathscr{I}_j^r (\mathbb{P})} } |S_I|$, the following result follows from Theorem \ref{S_I cardinality}:
\begin{corollary}\label{S_r cardinality}
		Let $I$ be an ideal in $\mathbb{P}$. Then the cardinality of $I$-sphere centered at  $ u \in \mathbb{Z}_{m}^N $  is 
	\begin{align*}
		| S_{(Pm,\pi)}(u,r) | =  \sum_{j=1}^{min\{i, n\}} \sum_{I \in {\mathscr{I}_{j}^r(\mathbb{P})} }\prod\limits_{i_t \in Max(I)^*} \delta(c_{i_t})  \prod\limits_{l \in {I \setminus Max(I)}} m^{k_l}.
	\end{align*}
	where $c_{i_t} = C_I(i_t)$ and $\delta(c_{i_t} ) =  \left\{
\begin{array}{ll}
	(2c_{i_t}+1)^{k_{i_t}} - (2c_{i_t}-1)^{k_{i_t}} & \hbox{if $c_{i_t} \neq \lfloor \frac{m}{2} \rfloor $ } \\
	m^{k_{i_t}} - (m-1)^{k_{i_t}}  & \hbox{if $c_{i_t} = \lfloor \frac{m}{2} \rfloor $.}
\end{array} \right. $ 
\end{corollary}
  \par  From Proposition \ref{runion ball}, we have
\begin{equation*}
	 B_{(Pm,\pi)}(u,r)= \bigcup\limits_{I \in \mathscr{I}^r(\mathbb{P})}B_{I}(u) = \bigcup\limits_{r=1}^{n} \bigcup\limits_{I \in \mathscr{I}^r(\mathbb{P})} S_{I}(u).
\end{equation*}
 Since cardinality of an $r$-ball does not depend on its centre,  $	| B_{(Pm,\pi)}(0, r)| = 1 + \sum\limits_{j=1}^{r} |S_{(Pm,\pi)}(0, j) |= 1 + \sum\limits_{i=1}^{r}\sum\limits_{j=1}^{min\{i, n\}} \sum\limits_{I \in \mathscr{I}_{j}^i(\mathbb{P})} |S_I| $.  Thus, we have the following:
\begin{corollary}
	 The cardinality of $(Pm,\pi)$-ball centered at $ u \in \mathbb{Z}_{m}^N $ with radius $r$ is 
	\begin{align*}
		| B_{(Pm,\pi)}(u, r) | = 1 + \sum_{i=1}^{r} \sum_{j=1}^{min\{i, n\}} \sum_{I \in {\mathscr{I}_{j}^i} (\mathbb{P})}  \prod_{i_t \in Max(I)^*} \delta(c_{i_t} ) \prod_{l \in {I^* \setminus Max(I)^*}} m^{k_l}.
	\end{align*}
	where $c_{i_t} = C_I(i_t)$ and $\delta(c_{i_t} ) =  \left\{
\begin{array}{ll}
	(2c_{i_t}+1)^{k_{i_t}} - (2c_{i_t}-1)^{k_{i_t}} & \hbox{if $c_{i_t} \neq \lfloor \frac{m}{2} \rfloor $ } \\
	m^{k_{i_t}} - (m-1)^{k_{i_t}}  & \hbox{if $c_{i_t} = \lfloor \frac{m}{2} \rfloor $.}
\end{array} \right. $ 
\end{corollary}
\par 	 Now, the $(Pm,\pi)$-weight distribution of $ \mathbb{Z}_m^{N}$ is obtained as follows: 
	For each $1 \leq r \leq n\big\lfloor\frac{ m }{2}\big \rfloor$, let $A_r = \{u= u_{1} \oplus u_{2} \oplus \ldots \oplus u_{n} \in  \mathbb{Z}_m^{N}  : w_{(Pm,\pi)}(u)=r  \}$ and $A_0 = \{ \bar{0}\}$. Then, $u \in A_r$ iff $ u \in S_r$. Thus, we have:
	\begin{proposition}\label{A_r Cardinality pomset block}
	For any $1 \leq r \leq n \big\lfloor \frac{ m }{2} \big\rfloor $,  the number of $N$-tuples  $x \in \mathbb{Z}_{m}^N$ having $w_{(Pm,\pi)}(x)= r$  is $|A_r|=| S_r |$.
\end{proposition}
\subsection{Ideals with full count} 
In \cite{gr}, it was shown that unlike the results in poset space \cite{hkmdspc}, $I$-balls are no more linear subspaces of $\mathbb{Z}_m^n$ in the pomset space when the ideal $I$ is with a partial count. Thus, there is a need to investigate the properties of $I$-balls with respect to pomset block space as well.
  The following result is a block version to the result in  [\cite{gr} ref. Proposition $3$] and its proof follows the similar pattern. 
\begin{proposition} \label{full count theorem}
	Let $\mathbb{P} $ be a pomset on a regular multiset $ M$ and $\tilde{\mathbb{P} }$ be the dual pomset of $\mathbb{P} $. If $I$ is an ideal with full count in $\mathbb{P}$, then 
	\begin{enumerate}[label=(\alph*)]
		\item $B_I$ is a submodule of $\mathbb{Z}_m^N$ and $|B_I| = m ^{\sum\limits_{i \in I^*} k_i}$.
		\item  For $u \in \mathbb{Z}_m^N$, $B_{I}(u)$ is the coset of $B_I$ containing $u$, ie. $B_I(u) = u + B_I$. 
		\item   For $u,v \in \mathbb{Z}_m^N$, the two $I$-balls $B_I(u) $ and $B_I(v)$ are either identical or disjoint. Moreover,
		$	B_I(u)  =  B_I(v)  \text{ if and only if  } supp_{(Pm, \pi)}{(u-v)} \subseteq I$.
		\item\label{4d}  $B_{I,~\mathbb{P}}^\perp (0)= B_{I^c, \tilde{\mathbb{P}}}(0)$.
	\end{enumerate} 
\end{proposition}
	Hence,	$(\mathbb{Z}_m^N,~d_{(Pm,\pi)})$-space can be partitioned into $I$-balls for  an ideal  $ I $ with full count.
	If $t$ be the number of distinct $I$-balls in the partition of $\mathbb{Z}_m^N$, then $t = m^{N -\sum\limits_{i \in I^*}k_i}$. If $k_i=k$ for each $i \in [n]$, then $t = m^{(n - |I^*|)k}$.
 \begin{definition}
 	A $(Pm,\pi)$-code $\mathbb{C} \subseteq \mathbb{Z}_m^N$  is said to be an $I$-perfect code if the $I$-balls centered at the codewords of $\mathbb{C}$ are  pairwise disjoint and their union is $\mathbb{Z}_m^N$.
 \end{definition}
\begin{theorem} \label{I-perfect full pomset code}
	If $I$ is an ideal with full count in $\mathbb{P}$ then there exists an $I$-perfect pomset block code.
\end{theorem}
\begin{proof}
	By Proposition \ref{full count theorem},  the space $\mathbb{Z}_m^N$ can be partitioned into $I$-balls. The number of $I$-balls is  $ m^{N -\sum\limits_{i \in I^*}k_i}$. Then the set of $N$-tuples formed by picking exactly one $N$-tuple from each $I$-ball forms an $I$-perfect $(Pm,\pi)$-code $\mathbb{C}$ of length $N$ over  $\mathbb{Z}_m$ and  $ |\mathbb{C} | = m^{N -\sum\limits_{i \in I^*}k_i}$ with  $d_{(Pm,\pi)} (\mathbb{C}) > |I|$.  	
\end{proof}
\begin{lemma} 
	Let $I$ be an ideal with full count. If $\mathbb{C}$ is an $I$-perfect pomset block code of cardinality $m^s$ and  $k_i = k$ for all $i$ then $|I^*|= \frac{kn-s}{k}$. Moreover if $k = s$ then $|I^*|= n - 1$.
\end{lemma}
\begin{proof}
	Since $\mathbb{C}$ is an $I$-perfect pomset block code, $|\mathbb{C}||B_I| = m^{N}$ and $|B_I| = m^{k|I^*|}$ imply  that  $|I^*|= \frac{N-s}{k}$.
\end{proof}
\begin{theorem} \label{C I to C I perc code}
	Let $\mathbb{C} \subseteq \mathbb{Z}_m^N $ be a linear pomset block code of cardinality $m^s$ and $I$ be an ideal with full count in $\mathbb{P}$. Then, $\mathbb{C}$ is $I$-perfect if and only if $\mathbb{C^\perp}$ is $I^c$-perfect where $I^c$ is an ideal in $\tilde{\mathbb{P}}$. 
\end{theorem}
\begin{proof}
	Suppose that $\mathbb{C}$ is $I$ perfect. Let $B_{I,\mathbb{P}}$  be an $I$-ball centered at $0$ in $\mathbb{C}$, then $\sum\limits_{i \in I^*} k_i = N -s$. Since $I$ is an ideal with full count in $\mathbb{P}$, then $I^c$ is an ideal with a full count in $\tilde{\mathbb{P}}$. So, $\sum\limits_{i \in I^{c*}} k_i = s$. Let $B_{I^c,\tilde{\mathbb{P}}}$ be an $I^c$-ball centered  at $0$ in $\mathbb{C^\perp}$. Then $|B_{I^c,\tilde{\mathbb{P}}}| = m^s$. We have to prove that   $\mathbb{C^\perp}$ is $I^c$-perfect. It is enough to show that $|B_{I^c,\tilde{\mathbb{P}}} \cap \mathbb{C^{\perp}}| = 1$. Suppose that $  u \in B_{I^c,\tilde{\mathbb{P}} }\cap \mathbb{C^{\perp}}$. From Proposition \ref{full count theorem} \ref{4d}, $ u \in B_{I,\mathbb{P}}^{\perp}$ and thus $u \cdot x=0 ~\forall ~x \in B_{I,~\mathbb{P}}$. Since $\mathbb{C}$ is $I$-perfect each $y \in \mathbb{Z}_m^N $ occurs exactly in one $I$-ball, say $B_{I,\mathbb{P}}(c)$ for some $c \in \mathbb{C}$, so that  $y = c + x$ for some $x \in B_{I,\mathbb{P}}$. Thus, $u\cdot y=0$ for every $y \in \mathbb{Z}_m^N $ which means $u=0$. Hence,  $|B_{I^c,\tilde{\mathbb{P}}} \cap  \mathbb{C^{\perp}}| = 1$.  The converse follows in a similar fashion.
\end{proof}
\subsection{Ideals with a partial count}\label{section33}
Now, for an ideal $I$ with partial count, $I$-balls may not hold all the properties arrived at so far. Consider an ideal $I$ with partial count in the pomset $\mathbb{P}$. There exist an $ i \in I^* $  such that $0 < C_I(i) < \lfloor \frac{m}{2} \rfloor$. 
Let $A_p= \{i_1,i_2,\ldots,i_{I_p}\}$ be the  set  of  those  elements in $I^*$ that have partial count in $I$ and $A_f= \{j_1,j_2,\ldots,j_{I_f}\}$ be the set of those elements in $I^*$ that have full count in $I$. Thus, $I^* = A_p \cup A_f$ with $|A_p|=I_p$ and $|A_f|=I_f$. Let $c_{\ell}$ be the count of $\ell$ in $I $. Clearly, $1 \leq c_{i_s}   \leq  \lfloor \frac{m}{2}  \rfloor -1$ when $1 \leq s \leq  I_p$ and $c_{j_k} = \lfloor \frac{m}{2}  \rfloor$  for all $1 \leq k \leq  I_f$. 
Then,  the  $I$-ball $B_I$ consists of $N$-tuples of the form $v_{1} \oplus v_{2} \oplus \ldots \oplus v_{n}  \in \mathbb{Z}_{m}^N $ where $ v_i = (v_{i_1},v_{i_2},\ldots,v_{i_{k_i}}) \in \mathbb{Z}_{m}^{k_i} \text{ such that } v_{i_s} \in \{0,\pm1,\pm 2,\ldots,\pm c_{i_s} \} \text{ if } i \in A_p, ~ v_{i_s} \in  \mathbb{Z}_{m} \text{ if } i \in A_f \text{ and } v_{i_s} = 0 \text{ if } i \notin I^*$. 
Clearly,	$B_I$ is not a group with respect to addition of $N$-tuples and hence it is not a submodule of $\mathbb{Z}_m^N$.  Moreover, the following is hold concerning the cardinality of $B_I$:
 \begin{proposition}
 If	$I$ is an ideal with partial count in the pomset $\mathbb{P}$ then the cardinality of
 	 $I$-ball is 
 	 $
 	 |	B_I | =   m^{\sum\limits_{j \in A_{f}}k_{j} }\prod\limits_{i \in A_p}(1 + 2c_{i})^{k_{i}}$.	 
 \end{proposition} 
\begin{remark}
	Let $k_i=k$ $\forall$ $i$ and $I$ be an ideal with partial count in the pomset $\mathbb{P}$. Then the cardinality of	$I$-ball is 
	$	|	B_I | = m^{kI_f} \prod\limits_{i \in A_p}(1 + 2c_{i})^{k} $.	 
\end{remark} 
Since the $I$-balls centered at any two vectors are either disjoint or identical for an ideal $I$ with a full count, cosets of $B_I$ can partition the space $\mathbb{Z}_{m}^N$ into $I$-balls and one can determine $I$-perfect pomset block codes. But, in the case of an ideal $I$ with partial count, the $I$-balls need not be disjoint from one another.    
Translates of $B_I$ for such an ideal $I$ need not be disjoint always. We will identify if there exists any representative vectors $u,v \in \mathbb{Z}_{m}^N$ such that the translates $u+B_I$ and $v+B_I$ are disjoint, in the line of \cite{gr}. 
\subsection{Partition of  $\mathbb{Z}_m ^{k}$}\label{section34}
 Some results from \cite{gr} are restated here which will help to find the partition of $ \mathbb{Z}_m^k$ for any positive integer $k $.
  Let $t$ be an integer such that  $0 \leq t \leq \lfloor \frac{m}{2} \rfloor - 1$. 
Let $ S= \{a  \in \mathbb{Z}_m :  w_L{(a)} \leq t \}$ and $S' = \{a \in \mathbb{Z}_m :  w_L{(a)} > t \}$.  Thus, $S=\{0,\pm 1, \pm 2,\ldots,\pm t \}$, $|S| = 2t+1$ ($3 \leq 2 t + 1 \leq m-1 $), 
$S'=\{ t+1,  t+2,\ldots, m-t-1 \}$, $S \cap S'= \{\}$ and  $|S'| = m-2t-1$. If $a \in S \smallsetminus \{0\}$ then $a+S =-a+S$. Moreover, $2t+1 \mid  m-2t-1 $ if and only if $2t+1 \mid m$; that is,  $ |S| \mid |S'|$ if and only if $|S| \mid  m$.
\begin{proposition}[\cite{gr}] \label{S and a plus S}
	$S$ and $a +S$ are neither disjoint nor identical for any $a \in S \smallsetminus \{0\}$.
\end{proposition} 
 $S$ can be written as disjoint union of two subsets  $S_1 = \{0, 1,\ldots , t\}$, and $ S_2 = \{m-t, m-t+1,\ldots ,m-1 \}$. If $ 2t + 1 \mid m - 2t -1 $ and $\ell = \frac{m - 2t -1 }{2t+1}$, then $ m =(\ell + 1)(2t + 1)$. Clearly, $b \in S'$ iff $t < b < m-t$.  Hence,   $j(2t+1)+a \in S'$ for all $ a \in S$ whenever $ 1 \leq j  \leq \ell$.  
\begin{lemma}[\cite{gr}] \label{Translation of S}
	Let  $2t+1 \mid m-2t-1$ and $\ell = \frac{m-2t-1}{2t+1}$. Then the following hold: 
	\begin{enumerate}[label=(\roman*)]
		\item\label{2a}  For each $ i = 1,2, \ldots , \ell$, we have $a + i(2t + 1) (mod ~m) \notin S$ for any $a \in S$.
		\item\label{2b} For $i\neq j$, the translates $i(2t+1) + S$ and $ j (2t+1) + S$ are disjoint where $1 \leq i$, $j \leq \ell$. Moreover, $
		\dot{\bigcup\limits_{0\leq i \leq \ell }} i(2t+1)+S = \mathbb{Z}_m $.
		\item\label{2c}   $\langle 2t+1 \rangle$ is a submodule of $\mathbb{Z}_m$.
	\end{enumerate}
\end{lemma}
Hence,   whenever $2t+1 \mid m-2t-1$, there exist translates of $S$ that  partition  $\mathbb{Z}_m$.	
If $2t+1 \nmid   m-2t-1$ then the translates of $S$ will not form a partition of $\mathbb{Z}_m$. For further details, one can refer \cite{gr}.
\par 
Considering the findings above, we now define  $S^k= \{s = (s_{1},s_{2},\ldots,s_{k})\in \mathbb{Z}_m^{k} : s_{j} \in S \}$ and $S'^k = \{s' = (s'_{1},s'_{2},\ldots,s'_{k})\in \mathbb{Z}_m^{k} : s'_{j} \in S' \}$. As $S$ is not a subgroup of $ \mathbb{Z}_m$,  $S^k$ is not a subgroup of $ \mathbb{Z}_m^{k}$. Clearly, $|S^k|= (2t+1)^k$ and $|S'^k|= (m-2t-1)^k$.
\begin{proposition} 
	$S^k$ and $ x +S^k$ are neither disjoint nor identical for any $ x \in S^k \smallsetminus \{0\}$.
\end{proposition}
\begin{proof}
   Let $x = (x_{1}, x_{2},\ldots,x_{k}) \in S^k \smallsetminus \{0\}$, where $x_{i} \in S $.  Let 
   \begin{align*}
   	y =  (y_{1}, y_{2},\ldots,y_{k}) \text{ where }
   	y_i  =  \left\{
   	\begin{array}{ll}
   		- ( x_i - 1 ) + t & \hbox{if $1 \leq x_i \leq t$   } \\
   		- ( x_i + 1 )  - t & \hbox{if $-t \leq x_i \leq -1 $.}
   	\end{array} \right.
   \end{align*}
 Clearly,  $ y_i \in S$ and  $ x_i + y_i \in S'$ for each $i$. Thus,  $ y  \in  S^k $ and  $x + y  \in  x + S^k $. Moreover,  $ x + y  \in S'^k $ (as $ x_i + y_i  \in S' $) which means $ x + y  \notin S^k $.   Hence,  $S^k$ and $ x +S^k$  are not identical. Also, $x \in S^k $  and  $x \in  x + S^k $. Thus,	$S^k$ and $ x +S^k$ are neither disjoint nor identical for any $ x \in S^k \smallsetminus \{0\}$.
\end{proof}
  We have  $j(2t+1)+a \in S'$ for all $0 \neq a \in S$, $ 1 \leq j  \leq \ell $ whenever $2t+1 \mid m-2t-1$ and $ \ell = \frac{m - 2t -1 }{2t+1}$. 
  Now, we define $T =\{j(2t+1) : 0 \leq j \leq \ell \}$  and $T^{k} = \{v = (v_{1}, v_{2},\ldots, v_{k})\in \mathbb{Z}_m^{k} : v_{i} \in T  \}$. Clearly, $|T^k| = (\ell+1)^k= (\frac{m}{2t+1})^k$.
\begin{lemma}\label{translation of B}
	Let  $2t+1 \mid m-2t-1$ and $\ell = \frac{m-2t-1}{2t+1}$. Then the following hold: 
	\begin{enumerate}[label=(\roman*)]
		\item For each  $v \in T^{k} \setminus \{0\}$,  $S^k$  and $v + S^k$ are disjoint.
		\item\label{l32} For $v \neq v' \in  T^{k} \setminus \{0\}$, the translates $v  + S^k$ and $ v' + S^k$ are disjoint. Moreover,
		$ \dot{\bigcup\limits_{v \in T^{k} }}	v + S^k = \mathbb{Z}_m^{k}$.
		\item  $\langle v= (v_{1}, v_{2},\ldots, v_{k}) \rangle$ is a submodule of $\mathbb{Z}_m ^{k}$, where $v_{j} \in  \langle 2t+1 \rangle$ for all $1 \leq j \leq k$.
	\end{enumerate}
	Hence, if $2t+1 \mid m-2t-1$ then the translates of $S^k$ form a partition of $\mathbb{Z}_m ^{k}$.
\end{lemma}
\begin{proof}
	\begin{enumerate}[label=(\roman*)]
		\item  Let $v  \in T^{k} \setminus \{0\}$ and $x \in S^k$. Suppose that $x_j$ and $v_j $ be the $j^{th}$ coordinate of $x$ and $v$ respectively. Since $v_j \in T$, from Lemma \ref{Translation of S} \ref{2a},  $v_{j}+x_{j} \in S'$ for any $x_j $ in $S$. Thus,  $v + x \notin S^k$ for any $x \in S^k$. Hence,  $S^k$  and $v + S^k$ are disjoint.
		\item  Suppose that the translates $v + S^k$ and $ v' + S^k$ are not disjoint for some $v \neq v' \in  T^{k} \setminus \{0\}$. Now, $v$ and $v'$ will differ in at least one position say $i_0$ where $v_{i_0} = j(2t+1) \neq v'_{i_0}  = p(2t+1)$  for some  $j \neq p$. Then, there exist  a $u \in (v + S^k)  \cap  (v' + S^k)$, so that $ u = v + x $ and $ u = v' + y$ for some $x, y \in S^k$. $ (v_{1}, v_{2},\ldots, v_{k}) + (x_{1},x_{2},\ldots,x_{k}) = (v'_{1}, v'_{2},\ldots, v'_{k}) + (y_{1}, y_{2},\ldots,y_{k})$ implies that  $ v_{i} + x_{i} =  v'_{i} + y_{i} $ for all  $1 \leq i \leq k$. Then,  $v_{i_0} + x_{i_0} \in v'_{i_0} + S$  and $v'_{i_0} + y_{i_0} \in v_{i_0} + S$. Thus, $j(2t+1) + S $ and $p(2t+1) + S $ are not disjoint which contradicts the Lemma \ref{Translation of S} \ref{2b}.
		\item Let $H= \langle v= (v_{1}, v_{2},\ldots, v_{k}) \rangle$ where $v_{j} \in  \langle 2t+1 \rangle$ for all $1 \leq j \leq k$. From Lemma \ref{Translation of S} \ref{2c}, we have $\langle 2t+1 \rangle$ is a submodule of $\mathbb{Z}_m$.  Consider two vectors $u , w \in H $ where $u+w=(u_{1}, u_{2},\ldots, u_{k}) + (w_{1}, w_{2},\ldots, w_{k})=(u_{1}+w_{1}, u_{2}+w_{2},\ldots, u_{k}+w_{k}) $. Since $u_j , w_j \in \langle 2t+1 \rangle$, $u_j + w_j \in \langle 2t+1 \rangle$ $\forall$ $j$ and thus, $u+w \in H$. For $a \in \mathbb{Z}_m$,  $au_j \in \langle 2t+1 \rangle$ and $au \in H$. Hence, $H$ is a submodule of $\mathbb{Z}_m^k$.
	\end{enumerate}
\end{proof}
\begin{remark}
	If $2t+1 \nmid m-2t-1$ then the translates of $S^k$ will not form a partition of $\mathbb{Z}_m^{k}$.
\end{remark}
  Now, we will find the $I$-perfect pomset block code $\mathbb{C}$ of length $N$ for an ideal $I$ with a partial count. 
  Refer the notations $A_p$ and $A_f$ in Section \ref{section33} with regard to the ideal $I$. We denote the partial count of each $i \in A_p$ in $I$ as $t_{i}$. Clearly, $1 \leq t_{i_s}   \leq  \lfloor \frac{m}{2}  \rfloor -1$ for $1 \leq s \leq  I_p$.  Suppose that $2t_i+1 \mid m-2t_i-1$ for an $i \in \{i_1, i_2, \ldots , i_{I_p}\}$ and $ \ell_i = \frac{m - 2t_i -1 }{2t_i+1}$. The foregoing results connected with $S'^k$ and $T^k$ for a specified $t$, can now be used for each $t_i$.
  
  \par Let  $D_i = \{(\bar{0},\ldots,\bar{0},v_i,\bar{0},\ldots,\bar{0}) \in \mathbb{Z}_m^N :  v_i \in T^{k_i} \}$. 
   Let $v=(\bar{0},\ldots,\bar{0},v_i,\bar{0},\ldots,\bar{0})$ and $v'=(\bar{0},\ldots,\bar{0},v'_i,\bar{0},\ldots,\bar{0}) $ be the two distinct $N$-tuples in $D_i$. Then $B_I(v)=v+B_I$ and $B_I(v')=v'+B_I$ where 
  	\begin{equation*}
  	B_I = \left\{
  	\begin{array}{ll}
  		x_{1} \oplus x_{2} \oplus \ldots \oplus x_{n}  \in \mathbb{Z}_{m}^N  :   x_i = \left\{ 
  		\begin{array}{ll}
  			x_i \in \mathbb{Z}_m^{k_i}, & \text{for} ~i \in I^*-A_p \\
  			x_i \in S^{k_i}, &  \text{for}  \ i \in A_p \\
  			\bar{0}, &  \text{for}  \ i \in [n]-I^*
  		\end{array}	\right\} \end{array} \right\}
  \end{equation*}
  Since $v_i \neq v'_i$, $v_i+S^{k_i}$ and $v'_i+S^{k_i}$ are disjoint by  Lemma \ref{translation of B} \ref{l32}. Thus, the $I$-balls centered at distinct $N$-tuples $v$ and $v'$ in $D_i$ are disjoint.
  \par 
  This allows one to construct a set of $N$-tuples such that the $I$-balls centered at vectors in this set are pairwise disjoint. For this, the ideal $I$ with a partial count must be such that   $2t_{i} + 1 \mid m-2t_{i}-1$ for each $i \in A_p =\{i_1, i_2, \ldots , i_{I_p}\} $. Now, define  
 $D_{A_p} = D_{i_1}  \oplus D_{i_2} \oplus \ldots \oplus D_{i_{I_p}} $
  	and $D_{[n]-I^*} = \{v_{1}  \oplus v_{2} \oplus \ldots \oplus v_{n} \in \mathbb{Z}_{m}^N : v_i = 0 \text{ when } i \in I^* \} $.
  	Note that, for a $v \in D_{A_p}$, $v_i = 0 ~\forall~ i \notin A_p $ and for a  $v \in D_{[n]-I^*}$, $\ v_i = 0$ $\forall$ $i \notin [n] - I^* $.
  \par 
    As $I$-balls centered at the vectors in $D_i$ are disjoint, the $I$-balls centered at the vectors in $ D_{A_p}$ are also  disjoint. 
   \begin{theorem}\label{D_I are disjoint } 
   	Let $I$ be an ideal with partial count and $t_j$ be the partial count of $j$ in $I $. If $2t_{j} + 1 \mid m-2t_{j}-1$ for each $j \in A_p=\{i_1, i_2, \ldots , i_{I_p}\}  $, then the $I$-balls centered at the $N$-tuples of  $	D_{A_p} =$
   		\begin{equation*}
   	 \left\{
   		\begin{array}{ll}
   			v=v_{1}  \oplus v_{2} \oplus \ldots \oplus v_{n} \in \mathbb{Z}_{m}^N  : v_{i} \in \mathbb{Z}_m^{k_i} \ \text{and}  \ v_i  = \left\{ 
   			\begin{array}{ll}
   				0, & \text{for} ~i \notin A_p \\
   				v_i \in T^{k_i}, &  \text{for}  \ i \in A_p 
   			\end{array}	\right\} \end{array} \right\}
   	\end{equation*}
   	   are disjoint.
    Moreover, $ |D_{A_p}| = \prod\limits_{j \in A_p}(1 + \ell_j)^{k_{j}} 	$ where $\ell_j = \frac{m-2t_j-1}{2t_j+1}$.
   \end{theorem}
 \par Now, we can construct the $I$-perfect pomset block code by considering the direct sum of $D_{A_p}$ and $D_{[n]-I^*}$.
\begin{theorem}\label{I-perfect partial pomset code } 
	Let $I$ be an ideal with partial count and $t_j$ be the partial count of $j$ in $I $. If $2t_{j} + 1 \mid m-2t_{j}-1$ for each $j \in A_p =\{i_1, i_2, \ldots , i_{I_p}\} $, then 
	$D =$ 
	\begin{equation*}
  \small{ \left\{
		\begin{array}{ll}
		v=v_{1}  \oplus v_{2} \oplus \ldots \oplus v_{n} \in \mathbb{Z}_{m}^N  : v_{i} \in \mathbb{Z}_m^{k_i} \ \text{and}  \ v_i  = \left\{ 
			\begin{array}{ll}
				0, & \text{for} ~i \in I^*-A_p \\
				v_i \in T^{k_i}, &  \text{for}  \ i \in A_p \\
				v_i \in \mathbb{Z}_m^{k_i}, &  \text{for}  \ i \in [n]-I^*
			\end{array}	\right\} \end{array} \right\} }
	\end{equation*}
	is an $I$-perfect $(Pm,\pi)$-code of length $N$ over $\mathbb{Z}_m$. Moreover, $ |D| = m^{\sum\limits_{i \in {[n]-I^*}} k_{i} } \prod\limits_{j \in A_p}(1 + \ell_{j})^{k_{j}} $ where $\ell_j = \frac{m-2t_j-1}{2t_j+1}$.
\end{theorem}
\section{MDS and $I$-Perfect Codes with Chain Pomset}
 Throughout this section, $\mathbb{P}=(M, R)$ is considered to be a chain. Then  $|\mathscr{I}^{r}(\mathbb{P})| =1$ for each $ r \leq n \lfloor\frac{ m }{2} \rfloor$. For the ideal $I \in \mathscr{I}^r (\mathbb{P})$, $B_I (x) = B_{(Pm,\pi)}(x, r)$ for any $x \in \mathbb{Z}_m^N$. Moreover, the poset $P=([n],\preceq)$ induced by the pomset $\mathbb{P}$ is also a chain. Each ideal in $\mathbb{P}$ or in $\mathbb{P}$ has a  unique maximal element. Let $ c_{i_t} / i_t $ be the maximal element of  $ \langle supp_{(Pm,\pi)}(x) \rangle$. Then, 
  \begin{equation*}
  	w_{(Pm,\pi)}(x)= c_{i_t} + (| \langle {c_{i_t}/i_t} \rangle ^*|-1) \lfloor\frac{ m }{2} \rfloor.
  \end{equation*}
\par In this section, the Singleton bound and the packing radius for pomset block codes are obtained. Furthermore, we look into the connection between MDS pomset block codes and $I$-perfect ($r$-perfect) pomset block codes. 
\begin{lemma}
	Let $\mathbb{P}$ be a chain pomset. Then a code  $\mathbb{C} \subseteq \mathbb{Z}_m^N $ is an $r$-perfect $(Pm, \pi)$-code if and only if $\mathbb{C}$ is an $I$-perfect $(Pm, \pi)$-code for the ideal $I \in \mathscr{I}^r (\mathbb{P})$.
\end{lemma}
Let  $r = t \lfloor\frac{ m }{2}\rfloor+s$ where $s \in  \{1,2,\ldots, \lfloor\frac{ m }{2}\rfloor\}$ and $t \geq 0$. Let $B_{(Pm,\pi)} (x,r)$ and $B_{(P,\pi)} (x,r)$ denote the  $r$-balls centered at $x$ having radius $r$ with respect to the $(Pm,\pi)$-metric and $(P,\pi)$-metric respectively. If $y \in B_{(P,\pi)} (x,r)$, then $w_{(Pm,\pi)}(x-y) \leq t\lfloor\frac{ m }{2}\rfloor + s$ and $w_{(P,\pi)}(x-y) = |\langle supp_{\pi} (x-y) \rangle| \leq t+1$. Hence, $y \in B_{(P,\pi)} (x,t+1)$. Thus, we have
\begin{theorem}\label{Packing radius pwpi code lemma}
	Let $ x \in \mathbb{Z}_{m}^{N} $ and $w$ be a weight on $\mathbb{Z}_{m}$. If $r = t \lfloor\frac{ m }{2}\rfloor + s$,  $1 \leq s \leq   \lfloor\frac{ m }{2}\rfloor $ and $t \geq 0$,   then $B_{(P,\pi)} (x,r) \subseteq B_{(P,\pi)} (x,t+1)$. Furthermore, $B_{(P,\pi)} (x,r) = B_{(P,\pi)} (x,t)$ iff $r = t \lfloor\frac{ m }{2}\rfloor $.  
\end{theorem}
\begin{definition}
	The packing radius of a  code $\mathbb{C} $ with respect to any metric $d$ is the
	greatest integer $r$ such that the  $r$-balls centered at any two distinct codewords are disjoint. 
\end{definition}
\par  Let $R_{(Pm,\pi)} (\mathbb{C})$  and $R_{(P,\pi)} (\mathbb{C})$ denote the packing radius of the code $\mathbb{C} \subseteq \mathbb{Z}_{m}^{N} $ with respect to the $(Pm,\pi)$-metric and $(P,\pi)$-metric respectively.   Let $R_{(P,\pi)} (\mathbb{C}) = t$ and $R_s = (t-1) \lfloor\frac{ m }{2}\rfloor + s$ where $1 \leq s \leq   \lfloor\frac{ m }{2}\rfloor $. Then, from Theorem \ref{Packing radius pwpi code lemma}, $B_{(Pm,\pi)} (x,R_s) \subseteq B_{(P,\pi)} (x,t)$ for each positive integer
$ s \leq \lfloor\frac{ m }{2}\rfloor$. Hence, $R_{(Pm,\pi)} (\mathbb{C}) \geq R_s $ for every $0 < s \leq \lfloor\frac{ m }{2}\rfloor$.  If $s = \lfloor\frac{ m }{2}\rfloor$, then from Theorem \ref{Packing radius pwpi code lemma},  $ R_{(Pm,\pi)} (\mathbb{C}) = R_{\lfloor\frac{ m }{2}\rfloor} = t \lfloor\frac{ m }{2}\rfloor $. Thus, packing radius of any  $(Pm,\pi)$-block code $ \mathbb{C}$ is $R_{(Pm,\pi)} (\mathbb{C}) = \lfloor\frac{ m }{2}\rfloor R_{(P,\pi)} (\mathbb{C}) $. As seen in \cite{nrt classification} (Theorem 5), $R_{(P,\pi)} (\mathbb{C}) = d_{(P,\pi)} (\mathbb{C}) -1 $ when $P$ is  a chain. Thus, packing radius of any  $(Pm,\pi)$-block code $ \mathbb{C}$ is $R_{(Pm,\pi)} (\mathbb{C}) = \lfloor\frac{ m }{2}\rfloor ( d_{(P,\pi)} (\mathbb{C}) - 1 ) $.
\begin{corollary}[Packing radius]
	Packing radius of a  $(Pm,\pi)$-block code $ \mathbb{C} \subseteq  \mathbb{Z}_{m}^{N} $ is $R_{(Pm,\pi)} (\mathbb{C}) = \lfloor\frac{ m }{2}\rfloor ( d_{(P,\pi)} (\mathbb{C}) - 1 ) $.
\end{corollary}
 When $P$ is a chain, the Singleton bound \cite{bkdnsr} for a $(P,\pi)$-code $\mathbb{C} $  is $  \sum\limits_{j \in J} k_{j} \leq  N - \lceil log_{m}|\mathbb{C}| \rceil $ where $J$ is an ideal in $P$ with $ |J| = d_{(P,\pi)}(\mathbb{C}) - 1$. Since the poset $P$ is induced by the pomset $\mathbb{P}$, there exist an ideal $I \in \mathbb{P} $ with $ | I | \leq  d_{(Pm,\pi)}(\mathbb{C}) - 1 $ and $|I^*|=\big\lfloor \frac{d_{(Pm,\pi)} \mathbb{(C)}-1}{ \big\lfloor \frac{ m }{2} \big\rfloor }\big\rfloor $. From Proposition \ref{phg}, we have $|I^*| \leq |J|$ and  $I^* \subseteq J$ as $P$ is a chain. Thus,
$ \sum\limits_{i \in I^*} k_{i} \leq     \sum\limits_{j \in J} k_{j} $. Hence, we have:
\begin{theorem}[Singleton bound] \label{chain Singleton bound}
		Let $\mathbb{C}$  be a chain pomset block code of length $N$ over $\mathbb{Z}_m$ with minimum distance $d_{(Pm,\pi)}(\mathbb{C})$.  Then $
	\max\limits_{I}  \sum\limits_{i \in I^*} k_{i} \leq N - \lceil log_{m}|\mathbb{C}| \rceil$ where  the maximum is taken over all ideals $I$ such that $|I| \leq d_{(Pm,\pi)}\mathbb{C}-1$ and $|I^*| = \big\lfloor \frac{d_{(Pm,\pi)} (\mathbb{C})-1} { \big\lfloor \frac{ m }{2} \big\rfloor } \big\rfloor $.
\end{theorem}
\begin{corollary}\label{singbound case}
	Let $\mathbb{C}$  be a pomset block code of length $N$ over $\mathbb{Z}_m$. Then the following holds: 
	\begin{enumerate}[label=(\roman*)]	 
		\item\label{singbound case 1} if $k_i = k$ for all $i \in [n]$ then $\big\lfloor \frac{d_{(Pm,\pi)} \mathbb{(C)}-1}{\big\lfloor\frac{ m }{2}\big \rfloor} \big\rfloor   \leq n- {\frac{\lceil log_{m}|\mathbb{C}|\rceil }{k}}$. 
		\item  if $k_1\geq k_2 \geq \ldots \geq k_n$ then $ {r} k_{n}  \leq N - \lceil log_{m}|\mathbb{C}| \rceil$ and $rk_n \leq \sum\limits_{i=1}^n k_i \leq rk_1$.
		\item  if $k_i = 1$ for all $i \in [n]$, then the Singleton bound for a pomset block code becomes that for a pomset code; that is, $	\big\lfloor \frac{d_{(Pm,\pi)} \mathbb{(C)}-1}{\big\lfloor\frac{ m }{2}\big \rfloor} \big\rfloor  \leq n - \lceil log_{m}|\mathbb{C}| \rceil$.
	\end{enumerate}
\end{corollary}
\begin{definition}
	A pomset block code $\mathbb{C} \subseteq \mathbb{Z}_m^N $ of length $N$ over $\mathbb{Z}_m$ is said to be maximum distance separable (MDS) if it attains its Singleton bound.  
\end{definition}
\begin{theorem}
	Let $k_i = k $  $\forall$ $i \in [n]$ and $\mathbb{C} $  be a $(Pm,\pi)$-code of length $N$  over  $\mathbb{Z}_m$ with minimum distance $d_{(Pm,\pi)}(\mathbb{C})$.  If $\mathbb{C}$ is MDS then $
	{\big\lfloor\frac{ m }{2}\big \rfloor}(n - \frac{\lceil log_{m}|\mathbb{C}| \rceil}{k})    + 1  \leq  d_{(Pm,\pi)} \mathbb{(C)} \leq   {\big\lfloor\frac{ m }{2}\big \rfloor}(n - \frac{\lceil log_{m}|\mathbb{C}| \rceil}{k}+1) $. 
\end{theorem}
\begin{proof}
	Since $\mathbb{C}$ is  MDS  and $k_i = k $  $\forall$ $i \in [n]$, we have  
	 $	\big\lfloor \frac{d_{(Pm,\pi)} \mathbb{(C)}-1}{\big\lfloor\frac{ m }{2}\big \rfloor} \big\rfloor  = n  - \frac{\lceil log_{m}|\mathbb{C}| \rceil}{k} $. Thus, 	$	n - \frac{\lceil log_{m}|\mathbb{C}| \rceil}{k} \leq  \frac{d_{(Pm,\pi)} \mathbb{(C)}-1}{\big\lfloor\frac{ m }{2}\big \rfloor} < n - \frac{\lceil log_{m}|\mathbb{C}| \rceil}{k} +1 $. Hence, $
	{\big\lfloor\frac{ m }{2}\big \rfloor}(n - \frac{\lceil log_{m}|\mathbb{C}| \rceil}{k})    + 1 \leq  d_{(Pm,\pi)} \mathbb{(C)} \leq   {\big\lfloor\frac{ m }{2}\big \rfloor}(n - \frac{\lceil log_{m}|\mathbb{C}| \rceil}{k}+1)$.		
\end{proof}
\par Therefore, if $k_i = k $  $\forall$ $i \in [n]$ and  $\mathbb{C} $ is a $(Pm,\pi)$-code of length $N$ over $\mathbb{Z}_m$ with minimum distance $d_{(Pm,\pi)}(\mathbb{C})$, then $\mathbb{C}$ cannot be an MDS pomset code  whenever $	1 \leq d_{(Pm,\pi)} \mathbb{(C)} \leq 	{\big\lfloor\frac{ m }{2}\big \rfloor}(n - \frac{\lceil log_{m}|\mathbb{C}| \rceil}{k})  $ or $ d_{(Pm,\pi)} (\mathbb{C})  > {\big\lfloor\frac{ m }{2}\big \rfloor}(n - \frac{\lceil log_{m}|\mathbb{C}| \rceil}{k}+1) $.
	\par 
Now we will examine the maximum distance separability of codes with respect to $(Pm,\pi)$-metric and $(P,\pi)$-metric.
\begin{theorem}
	If $\mathbb{C} $  is an MDS code with respect to $(Pm,\pi)$-metric then $\mathbb{C} $  is an MDS code with respect to $(P,\pi)$-metric.
	\begin{proof}
	Suppose that $\mathbb{C} $  is  MDS  with respect to $(Pm,\pi)$-metric. Then $ \sum_{i \in I^*} k_{i}  = N - \lceil log_{m}|\mathbb{C}| \rceil $, where $| I^*|= \big\lfloor \frac{d_{(Pm,\pi)} \mathbb{(C)}-1}{\lfloor \frac{m}{2} \rfloor} \big\rfloor$. Since $ |I^*| \leq d_{(P,\pi)}(\mathbb{C}) - 1 $ by  Proposition \ref{phg}, we have 
	$ \sum_{i \in I^*} k_{i} \leq \sum_{i \in J} k_{i}$ for an ideal $J$ in $P$ with $|J|= d_{(P,\pi)}(\mathbb{C}) - 1 $ as $I^* \subseteq J$.  Thus,  $  N - \lceil log_{m}|\mathbb{C}| \rceil \leq  \sum_{i \in J} k_{i} $.
		Hence $\mathbb{C} $ is MDS with respect to $(P,\pi)$-metric.
	\end{proof}
\end{theorem}
  \begin{theorem}\label{fullcount MDS}
  	Let $k_i =k$ for each $i \in [n]$ and $\mathbb{C} $ be a $(Pm, \pi)$-code of length $N$ over $\mathbb{Z}_m$ with cardinality $m^s$ for some $s>0$. If $\mathbb{C} $ is  MDS then $\mathbb{C} $ is $I$-perfect for all $I \in \mathscr{I}^{\lfloor \frac{m}{2} \rfloor (n-\frac{s}{k} )} (\mathbb{P})$.
  \end{theorem}
\begin{proof}
	Let  $\mathbb{C} $ be an MDS $(Pm, \pi)$-code. As  $\mathbb{P}$ is a chain, $ | \mathscr{I}^{\lfloor \frac{m}{2} \rfloor (n-\frac{s}{k} )} (\mathbb{P}) | = 1$ and so $I \in \mathscr{I}^{\lfloor \frac{m}{2} \rfloor (n-\frac{s}{k} )} (\mathbb{P})$ must be an ideal with full count with $I^*= n-\frac{s}{k} $. As $\mathbb{Z}_m^N$ can be partitioned into $I$-balls, let $\ell$ be the number of $I$-balls in this partition. Then $\ell |B_I| = m^N$ which gives  $\ell = |\mathbb{C}|$ as $|B_I| = m ^{k |I^*|}$. Since $\mathbb{C} $ is  MDS,  $	\big\lfloor \frac{d_{(Pm,\pi)} \mathbb{(C)}-1}{\big\lfloor\frac{ m }{2}\big \rfloor} \big\rfloor  = n  - \frac{s}{k} $ by Corollary \ref{singbound case} \ref{singbound case 1} and thus $d_{(Pm,\pi)} \mathbb{(C)} > |I|$. Therefore, any two $I$-balls centered at distinct codewords of $\mathbb{C}$ must be disjoint and $|\mathbb{C}||B_I|= m^N$. Hence $\mathbb{C} $ is $I$-perfect.
\end{proof}
\begin{theorem} \label{conversefullcountmds}
	Let $k_i =k$ for each $i \in [n]$ and $\mathbb{C} \subseteq \mathbb{Z}_m^N $ be a $(Pm,\pi)$-code of length $N$ over $\mathbb{Z}_m$. If $\mathbb{C}$ is $I$-perfect for the ideal $I \in \mathscr{I}^{\lfloor \frac{m}{2} \rfloor (n-\frac{\lceil log_{m}|\mathbb{C}| \rceil}{k})} (\mathbb{P})$  then $\mathbb{C}$ is MDS. 
\end{theorem}
\begin{proof}
	If $\mathbb{C}$ is  $I$-perfect for the ideal $I \in \mathscr{I}^{\lfloor \frac{m}{2} \rfloor (n-\frac{\lceil log_{m}|\mathbb{C}| \rceil }{k} )} (\mathbb{P})$, then $\lfloor \frac{m}{2} \rfloor (n - \frac{\lceil log_{m}|\mathbb{C}| \rceil }{k} )$ is an integer and $ d_{(Pm,\pi)} (\mathbb{C})  > \lfloor \frac{m}{2} \rfloor (n - \frac{\lceil log_{m}|\mathbb{C}| \rceil }{k} ) $.
	Thus,  $  \big\lfloor \frac{d_{(Pm,\pi)} \mathbb{(C)}-1}{\lfloor \frac{m}{2} \rfloor} \big\rfloor \geq  n-\frac{\lceil log_{m}|\mathbb{C}| \rceil }{k} $. Hence $\mathbb{C}$ is MDS.
\end{proof}
\begin{theorem}\label{p11}
  Let $k_i =k$ for each $i \in [n]$. Then every $I$-perfect $(Pm,\pi)$-code over $\mathbb{Z}_m$ is MDS.
\end{theorem}
\begin{proof}
	Let $\mathbb{C}$ be an $I$-perfect  $(Pm,\pi)$-code. Then $d_{(Pm,\pi)}(\mathbb{C}) > |I|$. If  $I$ is an ideal with full count then   $|B_I|= m^{\sum\limits_{i \in I^*}k_i}$, $|\mathbb{C}||B_I| = m^{N}$ and thus, we have $\sum_{i \in I^*}k_i = N - \lceil log_{m}|\mathbb{C}| \rceil$. Hence $\mathbb{C}$ is an MDS block code. Now, if $I$ is an ideal  with partial count then $|I| = |I'| + t $ for some ideal $I'$ with full count and $0 < t \leq \lfloor \frac{m}{2} \rfloor -1$. Since $\mathbb{C}$ is an $I$-perfect code, $ |\mathbb{C}||B_I| = m^{N} $ which means $	|\mathbb{C}|(2t + 1)^{k}|B_{I'}| = m^{nk}$. Taking $log$ to the base $m$ both sides we get,
	$	\sum\limits_{j \in I'^*}k_{j} = nk - \log_m|\mathbb{C}| - k \log_m(2t+1)$ so that 	$	 |I'^*| = n - \frac{\log_m|\mathbb{C}|}{k} -  \log_m(2t+1)$. 
	Since $0 < t \leq \lfloor \frac{m}{2} \rfloor -1 $, we get $ 0 < \log_m (2t+1) <1$. Since $\log_m|\mathbb{C}| + k \log_m(2t+1)$ is an integer, then $\frac{\log_m|\mathbb{C}|}{k} +  \log_m(2t+1) = \lfloor \frac{\log_m|\mathbb{C}|}{k} +  \log_m(2t+1) \rfloor \leq  \lfloor \frac{\log_m|\mathbb{C}|}{k} \rfloor + \lfloor  \log_m(2t+1) \rfloor +1  =  \lfloor \frac{\log_m|\mathbb{C}|}{k} \rfloor +1 =  \lceil  \frac{\log_m|\mathbb{C}|}{k}  \rceil $. 
	As $d_{(Pm,\pi)}(\mathbb{C}) > |I|= |I'| + t$, we have $  \frac{d_{(Pm,\pi)} \mathbb{(C)}-1}{\lfloor \frac{m}{2} \rfloor}  > n-\frac{log_{m}|\mathbb{C}| }{k} -  \log_m(2t+1) + \frac{t-1}{\lfloor \frac{m}{2} \rfloor} \geq n-\frac{ log_{m}|\mathbb{C}|  }{k} -  \log_m(2t+1) \geq n- \lceil  \frac{\log_m|\mathbb{C}|}{k}  \rceil$. Thus,  $ \big\lfloor \frac{d_{(Pm,\pi)} \mathbb{(C)}-1}{\lfloor \frac{m}{2} \rfloor} \big\rfloor \geq  n- \lceil  \frac{\log_m|\mathbb{C}|}{k}  \rceil $.
	Hence, $\mathbb{C}$ is MDS. 
\end{proof}
For a chain pomset, the following is the duality theorem of an MDS $(Pm, \pi)$-code when all the blocks are of same length. 
\begin{theorem}[Duality Theorem]
 Let  $\tilde{\mathbb{P}}$ be the dual pomset of the chain $\mathbb{P}$ on $M$ and $\pi$ be a labeling of $[n]$ with $k_1=k_2=\ldots=k_n=k$. Let $\mathbb{C}$ be an $(Pm, \pi)$-code of length $N$ over $\mathbb{Z}_m$ with cardinality $m^s$ for some $s>0$. Then the following statements are equivalent:	
	\begin{enumerate}[label=(\roman*)]
		\item $\mathbb{C}$ is an MDS   $\mathbb{P}$-code.
		\item $\mathbb{C}$ is an $I$-perfect $\mathbb{P}$-code for all $I \in \mathscr{I}^{\lfloor \frac{m}{2} \rfloor (n-\frac{s}{k})} (\mathbb{P})$.
		\item $\mathbb{C}^\perp$ is an $I^c$-perfect $\tilde{\mathbb{P}}$-code  for all $I^c \in \mathscr{I}^{\lfloor \frac{m}{2} \rfloor (\frac{s}{k})} (\tilde{\mathbb{P}})$.
		\item $\mathbb{C}^\perp$ is an MDS $\tilde{\mathbb{P}}$-code of cardinality $m^{N-s}$.
	\end{enumerate}
\end{theorem}
	\begin{proof}
	The proof is straightforward from Theorem \ref{fullcount MDS} and Theorem \ref{conversefullcountmds}.
\end{proof}

\bibliographystyle{amsplain}

\end{document}